\documentclass[english]{cccconf}
\usepackage[comma,numbers,square,sort&compress]{natbib}
\usepackage{epstopdf}
\usepackage{amsmath,amsfonts,amsthm,bm}
\usepackage[capitalise]{cleveref}
\usepackage{xcolor}
\usepackage{algorithm}
\usepackage{algpseudocode}
\usepackage{verbatim}
\usepackage{fp,tikz,pgfplots}

\usepackage{lipsum} 

\usetikzlibrary{arrows,shapes,backgrounds,patterns,fadings,matrix,arrows,calc,
	intersections,decorations.markings,
	positioning,arrows.meta}
\usepgfplotslibrary{fillbetween}
\usepgfplotslibrary{statistics}
\pgfplotsset{width=5\columnwidth /5, compat = 1.13,
	height = 60\columnwidth /100, grid= major,
	legend cell align = left, ticklabel style = {font=\scriptsize},
	every axis label/.append style={font=\small},
	legend style = {font={\scriptsize}},title style={yshift=-7pt, font = \small} }

\theoremstyle{remark}
\newtheorem{assumption}{Assumption}
\newtheorem{theorem}{Theorem}

\newtheorem{lemma}{Lemma}
\newtheorem{remark}{Remark}
\newtheorem{definition}{Definition}

\DeclareMathOperator*{\argmin}{arg\,min}

\xdefinecolor{blueZ}{RGB}{0,77,163}
\xdefinecolor{redZ}{RGB}{232,61,117}

\definecolor{tab10_blue}{RGB}{33,62,151}
\definecolor{tab10_orange}{RGB}{255,172,37}
\definecolor{tab10_gray}{RGB}{102, 102, 102}
\definecolor{tab10_green}{RGB}{44,160,44}
\definecolor{tab10_red}{RGB}{214,39,40}

\begin{document}

\title{
	Kernel-based Learning for Safe Control of Discrete-Time Unknown Systems under Incomplete Observations
}

\author{Zewen Yang$^{*}$\aref{rki},
        Xiaobing Dai$^{*}$\aref{tum},
        Weijie Yang$^{\dagger}$\aref{tyust},
        Bahar İlgen\aref{rki},
        Aleksandar Anžel\aref{rki},
        Georges Hattab\aref{rki,fub}
        }

\affiliation[rki]{Robert Koch Institute,
        Berlin, Germany
        \email{$\{$yangz; ilgenb; anzela; hattabg$\}$@rki.de}}
\affiliation[tum]{Technical University of Munich, Munich, Germany
        \email{xiaobing.dai@tum.de}}
\affiliation[tyust]{Taiyuan University of Science and Technology, Taiyuan, China
        \email{yangweijie@stu.tyust.edu.cn}}
\affiliation[fub]{Freie Universität Berlin, Berlin, Germany
        }

\maketitle

\begin{abstract}
Safe control for dynamical systems is critical, yet the presence of unknown dynamics poses significant challenges. In this paper, we present a learning-based control approach for tracking control of a class of high-order systems, operating under the constraint of partially observable states. The uncertainties inherent within the systems are modeled by kernel ridge regression, leveraging the proposed strategic data acquisition approach with limited state measurements. To achieve accurate trajectory tracking, a state observer that seamlessly integrates with the control law is devised. The analysis of the guaranteed control performance is conducted using Lyapunov theory due to the deterministic prediction error bound of kernel ridge regression, ensuring the adaptability of the approach in safety-critical scenarios. To demonstrate the effectiveness of our proposed approach, numerical simulations are performed, underscoring its contributions to the advancement of control strategies.
\end{abstract}

\keywords{Learning-based control, kernel ridge regression, discrete-time systems, high-order systems, state observer}

% Please remove or comment out the following line if the footnote is not necessary
\footnotetext{$^{*}$Equal contribution. $^{\dagger}$Corresponding author.}
\footnotetext{This work has been financially supported by the Germany Federal Ministry of Health (BMG) under grant No. 2523DAT400 (project ``AI-assisted analysis and visualization of pandemic situations'' | AI-DAVis-PANDEMICS), and by the Federal Ministry of Education and Research of Germany in the programme of ``Souverän. Digital. Vernetzt.'' under joint project 6G-life with project identification number: 16KISK002, and by Youth Project of Shanxi Basic Research Program (202203021212314).}

\section{Introduction}
Dynamical systems control has garnered significant attention due to its versatility and profound impact across diverse domains, ranging from robotics and electrical engineering to building management~\cite{8529192,gao2021quasi,yin2023learning,wan2023semi}. Despite the prevalent adoption of mechanistic models, they fail to capture precise parameters and provide a comprehensive description of system dynamics~\cite{o2013making}. This discrepancy is further exacerbated by the omnipresence of uncertainties and varying environments.

In recent times, there has been a growing inclination toward integrating machine learning methodologies to address aforementioned challenges in the context of unknown system control \cite{brunton2022data}. Machine learning techniques facilitate the discernment of latent patterns and the derivation of models from data, thereby augmenting the control performance in the face of incomplete system knowledge. A majority of studies employ neural networks (NNs) to either identify the system dynamics \cite{patino2000neural} or fit adaptive controllers~\cite{campestriniDatadrivenModelReference2017,weiDiscretetimeContractionbasedControl2022}. However, the rigorous guaranteed prediction is inadequate due to finite and patchy features. Furthermore, the parametric models including NN models suffer from limited complexity and flexibility, requiring intensive training and extensive data to effectively estimate complex systems.

Conversely, non-parametric techniques, an alternative avenue, leverage the kernel trick to ascertain the inner product, thereby encapsulating an infinite-dimensional feature space~\cite{chiusoSystemIdentificationMachine2019}. Such methods are more adaptable and accurate in modeling scenarios with limited data availability. Moreover, kernel methods offer a distinctive advantage by providing theoretical error bounds~\cite{srinivasInformationTheoreticRegretBounds2012,maddalenaDeterministicErrorBounds2021} making significant contributions to the domain of safety-critical control scenarios~\cite{huang2024learning}. 
Instances of such applications include the utilization of kernel ridge regression (KRR) in model predictive control~\cite{boulkaibet2017new}, support vector regression in conjunction with PID control~\cite{nguyen2009sparse}, and cooperative control incorporating Gaussian process regression~\cite{yangDistributedLearningConsensus2021,9882336,yang2024cooperative}, etc. 
However, the applicability of the literature above is contingent upon the feasibility of obtaining a complete state measurement, which presents practical challenges in real-world applications. 

In response to the inherent challenge of incomplete system state measurements, observers are employed to estimate the complete state based on limited measurements. The field has explored the application of neural observer-based adaptive control~\cite{chakrabartySafeLearningbasedObservers2021, chenNeuralObserverLyapunov2023}, where the control performance is strongly related to the accuracy of the chosen features.
However, the error induced by imprecise and truncated features lacks analytical results, hindering its application in safety-critical scenarios~\cite{dai2024cooperative,yang2024trust,dai2024decentralized}.
While certain endeavors have been made to leverage kernel methods with state observation~\cite{buisson-fenetJointStateDynamics2021,li2022synchronous}, there remains a conspicuous gap in the literature considering the interplay between the learning-based control of dynamical systems and observers. 

In this study, we introduce a kernel-based control method for a class of high-order systems with unknown dynamics, especially those with partially observable states. 
Moreover, a tailored data acquisition strategy is proposed to facilitate the limitation of state measurements, enabling the application of KRR to effectively model the uncertainties within the discrete-time system. 
By incorporating the learning-based control law with a state observer, the analytical ultimate bounds for both tracking error and observation error are provided.
The numerical simulation confirms the effectiveness of the proposed approach, emphasizing a notable enhancement in tracking performance. 

\section{Preliminaries}
\subsection{System Description and Objective}
In this paper, we consider a class of high-order discrete-time nonlinear dynamical systems described as follows
\begin{align} \label{eqn_dynamics}
	&x_i(t_{k+1}) = x_i(t_k) + x_{i+1}(t_k) T, ~ \forall i = 1, \cdots, n-1, \\
	\label{eqn_dynamics2}&x_n(t_{k+1}) = f(\bm{x}(t_k)) + u(t_k), \\
	&y(t_k) = x_1(t_k) + v(t_k),
\end{align}
where $\bm{x}(t_k) = [x_1(t_k), \cdots, x_n(t_k)]^T \in \mathbb{X}\subset \mathbb{R}^n$ with $x_i(t_k) \in \mathbb{R}, \forall i = 1, \cdots, n$ denotes the system state within compact domain $\mathbb{X}$, $u(t_k) \in \mathbb{R}$ represents the system input and $y(t_k) \in \mathbb{R}$ is the measured output at any time instance $t_k \in \mathbb{R}_{\geq0}$ with $k\in \mathbb{N}$.
The time interval $T = t_{k+1}-t_k \in \mathbb{R}_+$ is fixed, and the measurement noise $v(t_k)$ is independently and identically bounded, i.e., $|v(t_k)| \le \bar{v}, \forall k \in \mathbb{N}$, with $\bar{v} \in \mathbb{R}_{\geq0}$.
Although the structure of the system is known, the nonlinear function $f(\cdot): \mathbb{X} \to \mathbb{R}$ is unknown but estimated by the kernel-based learning function resulting in $\mu(\cdot): \mathbb{X} \to \mathbb{R}$ illustrated in \cref{subsec_krr}.
Moreover, it satisfies the following assumption.
\begin{assumption} \label{assumption_RKHS_bounded}
	Given a symmetric positive kernel $\kappa(\cdot, \cdot): \mathbb{X} \!\times\! \mathbb{X} \!\to\! \mathbb{R}_+$, the function $f(\cdot)$ belongs to the unique corresponding reproducing kernel Hilbert space (RKHS) $\mathcal{H}_{\kappa}$ with defined inner-product $\langle \cdot, \cdot \rangle_{\kappa}$.
	Additionally, the RKHS norm of $f(\cdot)$ is bounded by $B \!\in\! \mathbb{R}_+$, i.e., $\| f \|_{\kappa} \!=\!\! \sqrt{\langle f, f \rangle_{\kappa}} \!\le\! B$.
\end{assumption}
\cref{assumption_RKHS_bounded} indicates the unknown function is in the form as $f(\cdot) \!=\! \sum_{j \!=\! 1}^{\infty} \alpha_j \kappa(\bm{x}_j, \cdot)$, where $\alpha_j \!\in\! \mathbb{R}, \forall j \!=\! 1, \!\cdots\!, \infty$ are the coefficients.
Despite a similar form as linear regression, the assumption implies the non-parametricity of the unknown function $f(\cdot)$ considering that $\kappa(\bm{x}_i, \cdot)$ only reflects the correlation between the two points, which requires less prior information and cover larger function classes than parametric models, e.g., neural networks.
Indeed, some kernel functions, e.g., square exponential kernel, have the property of universal approximation, which means they can approximate any continuous function with arbitrary accuracy.
Moreover, the RKHS norm $\| \!\cdot\! \|_{\kappa}$ indicates the smoothness of the function set, and therefore the existence of the upper bound can be regarded as the demand on $f(\cdot)$ to be Lipschitz.
Note that the Lipschitz continuity in the compact $\mathbb{X}$ is easy to achieve by only requiring the function $f(\cdot)$ to be continuous.
Therefore, \cref{assumption_RKHS_bounded} imposes no practical restrictions.
\looseness=-1

The control task is to track a desired trajectory with a form
\begin{align} \label{eqn_reference}
	&s_i(t_{k+1}) = s_i(t_k) + s_{i+1}(t_k) T, ~ \forall i = 1, \cdots, n-1,
\end{align}
and $s_n(t_{k+1}) = r(t_k)$, where the function $r(\cdot): \mathbb{R}_{0,+} \to \mathbb{R}$ is predefined and $s_i(t_k) \in \mathbb{R}, \forall i = 1, \cdots, n$. Specifically, the object is to let $\boldsymbol{x}(t_k) \to \bm{s}(t_k)$ with a bounded neighborhood holds as $t_k\to \infty$ for $\forall \boldsymbol{x}(t_0) \in \mathbb{X}$, where $\bm{s} = [s_1, \cdots, s_n]^T \in \mathbb{X}$ denotes the state of the reference. 

Moreover, to infer the unknown function $f(\cdot)$ in \eqref{eqn_dynamics2}, the function $\mu(\cdot)$ is employed for prediction using collected data set denoted as $\mathbb{D}$ satisfying the following assumption.
\begin{assumption} \label{assumption_dataset}
	The data set $\mathbb{D}$ is composed of $N \in \mathbb{N}$ data pairs $\{ \bm{x}^{(\iota)\!}, z^{(\iota)\!} \}$ for $\iota \!\in\! \{1, \!\cdots\!, N\}$, where $z^{(\iota)\!} \!=\! f(\bm{x}^{(\iota)\!}) \!+\! w^{(\iota)\!}$ for $\forall \iota \!=\! 1, \!\cdots\!, N$.
	The measurement noise $w^{(\iota)} \!\in\! \mathbb{R}$ is bounded by $\bar{w} \!\in\! \mathbb{R}_{0,\!+}$, i.e., $| w^{(\iota)} | \!\le\! \bar{w}, \forall \iota \!=\! 1, \!\cdots\!, N$.
\end{assumption}
While \cref{assumption_dataset} is a common assumption in machine learning requiring full input observation, obtaining $\bm{x}^{(\iota)}$ is practically hard for our system \eqref{eqn_dynamics} due to limited measurements as $y$.
Moreover, \cref{assumption_dataset} admits the bounded measurement noise on $z^{(\iota)}$ as $\bar{w}$, which is directly inherited from bounded noise $v$ in \eqref{eqn_dynamics}.
To tackle the challenge arising from the constrained measurement of $y$ and unspecified $\bar{w}$, we introduce a data collection strategy and analyze the noise propagation from $\bar{v}$ to $\bar{w}$ in \cref{subsection_data_collection}. 

\subsection{Kernel Ridge Regression}\label{subsec_krr}
Owing to the advantageous properties of balance between achieving accurate data fitting and effectively mitigating the impact of uninformative fluctuations, KRR emerges as a widely adopted non-parametric machine learning technique. Given a kernel function $\kappa(\cdot,\cdot)$ and a data set $\mathbb{D} = \{ \bm{x}^{(\iota)}, z^{(\iota)} \}_{\iota = 1, \cdots, N}$, KRR aims to find the optimal solution from the corresponding RKHS $\mathcal{H}_{\kappa}$ with the following optimization problem as
\begin{align} \label{eqn_KRR_optimization}
	\mu \!=\! \argmin_{h \in \mathcal{H}_{\kappa}} N^{-1} \sum\nolimits_{\iota \!=\! 1}^N \Big( z^{(\iota)} \!-\! h(\bm{x}^{(\iota)}) \Big)^2 \!+\! \bar{w}^2 \| h \|^2_{\kappa}
\end{align}
with the penalty coefficient $\bar{w} \in \mathbb{R}_+$ in \cref{assumption_dataset}.
The closed-form solution of \eqref{eqn_KRR_optimization} is derived as follows.
\begin{lemma}[Representer theorem] \label{lemma_representer_theorem}
	Given a kernel $\kappa(\cdot,\cdot)$ with corresponding RKHS $\mathcal{H}_{\kappa}$ and a data set $\mathbb{D} = \{ \bm{x}^{(\iota)}, z^{(\iota)} \}_{\iota = 1, \cdots, N}$, the solution of \eqref{eqn_KRR_optimization} is presented as 
	\begin{align} \label{eqn_representer_theorem}
		\mu(\cdot) = \bm{k}^T(\cdot) \bm{\alpha}^*   ~~ \text{with} ~~ \bm{\alpha}^* = (\bm{K} + N \bar{w}^2 \bm{I}_N)^{-1} \bm{z},
	\end{align}
	where $\bm{k}(\cdot) \!\!=\!\! [\kappa(\!\bm{x}^{(1)\!}, \cdot\!), \!\cdots\!, \kappa(\!\bm{x}^{(N)\!}, \cdot\!)]^T$, $\bm{z} \!=\! [z^{(1)}, \!\cdots\!, z^{(N)}]^T$ and $\bm{K} = [\kappa(\bm{x}^{(p)}, \bm{x}^{(q)})]_{p,q = 1, \cdots,N}$
\end{lemma}
The closed form of $\mu$ provided in \cref{lemma_representer_theorem} is truncated, i.e., only using the kernels $\kappa(\bm{x}^{(\iota)}, \cdot)$ related to the given data set $\mathbb{D}$.
Therefore, the evaluation of $\mu$ in \eqref{eqn_representer_theorem} at any point $\bm{x} \in \mathbb{X}$ is practically possible with time complexity $\mathcal{O}(N)$.
Moreover, the prediction accuracy of the obtained $\mu(\bm{x})$ is quantified by a defined power function as follows.
\begin{definition}
	The data-dependent power function is defined as $P^2(\bm{x}) = \kappa(\bm{x},\bm{x}) - \bm{k}^T(\bm{x}) (\bm{K} + N \bar{w}^2 \bm{I}_N)^{-1} \bm{k}(\bm{x})$ with positive definite $P(\cdot)$, i.e., $P(\cdot): \mathbb{X} \to \mathbb{R}_+$.
\end{definition}
The positivity of $P^2(\cdot)$ is shown in \cite{maddalena2021deterministic} by observing its Lagrange form as in \cite{wendland2004scattered}, ensuring the real-valued power function.
On the other hand, $P(\cdot)$ can be regarded as the marginal probability of the distribution at $\bm{x}$ conditioned by data set $\mathbb{D}$ by comparing its expression with posterior variance from Gaussian process regression.
With the power function, we are able to assess the deterministic prediction error bound, which is shown as follows.
\begin{lemma} \label{lemma_KRR_error_bound}
	Consider an unknown function $f(\cdot)$ satisfying \cref{assumption_RKHS_bounded} with kernel $\kappa$ and RKHS norm bound $B$, which is predicted using KRR solving \eqref{eqn_KRR_optimization} with data set under \cref{assumption_dataset}.
	Then, the prediction error is bounded as
	\begin{align} \label{eqn_KRR_error_bound}
		| \mu(\bm{x}) \!-\! f(\bm{x}) | \!\le\! \beta P(\bm{x}), && \forall \bm{x} \in \mathbb{X}
	\end{align}
	with $\beta \!=\! \sqrt{B^2 \!-\! \bm{z}^T (\bm{K} \!+\! N \bar{w}^2 \bm{I}_N)^{-1} \bm{z} \!+\! 1}$.
\end{lemma}
\begin{proof}
	Define the data-dependent kernel function as $\kappa_{\mathbb{D}}(\bm{x}, \bm{x}') \!=\! \kappa(\bm{x}, \bm{x}') \!-\! \bm{k}^T(\bm{x}) (\bm{K} \!+\! N \bar{w}^2 \bm{I}_N)^{-1} \bm{k}(\bm{x}')$ such that the RKHS norm w.r.t $\kappa_{\mathbb{D}}$ following \cite{srinivas2012information} is written as
	\begin{align} \label{eqn_error_norm_in_kD}
		\| \mu \!-\! f \|^2_{\kappa_{\mathbb{D}}} \!=\! \| \mu \!-\! f \|^2_{\kappa} \!+\! (N \bar{w}^2)^{-1} \| \bm{\mu}_{\mathbb{D}} \!-\! \bm{f}_{\mathbb{D}} \|^2,
	\end{align}
	implying $\kappa_{\mathbb{D}}$ is also a kernel function with $\mathcal{H}_{\kappa_{\mathbb{D}}} \!=\! \mathcal{H}_{\kappa}$ shown in \cite{srinivas2012information}.
	The concatenated prediction and value of unknown function are expressed as $\bm{\mu}_{\mathbb{D}} \!=\! [\mu(\bm{x}^{(1)}), \!\cdots\!, \mu(\bm{x}^{(N)})]^T$ and $\bm{f}_{\mathbb{D}} \!=\! [f(\bm{x}^{(1)}), \!\cdots\!, f(\bm{x}^{(N)})]^T$, respectively.
	Moreover, apply the expression of $f$ and $\mu$ from \cref{assumption_RKHS_bounded} and \eqref{eqn_representer_theorem} respectively and consider
	\begin{align}
		\| \mu - f \|^2_{\kappa} &= \| f \|^2_{\kappa}  + \| \mu \|^2_{\kappa} - 2 \langle \mu , f \rangle_{\kappa} \nonumber \\
		&= \| f \|^2_{\kappa}  + (\bm{\alpha}^*)^T \bm{K} \bm{\alpha}^* - 2 \bm{f}_{\mathbb{D}}^T \bm{\alpha}^*, \\
		\bm{\mu}_{\mathbb{D}} - \bm{z} &= - N \bar{w}^2 \bm{\alpha}^*,
	\end{align}
	then the RKHS norm in \eqref{eqn_error_norm_in_kD} is further bounded as
	\begin{align}
		\| \mu \!-\! f \|^2_{\kappa_{\mathbb{D}}} \!\! =& \| f \|^2_{\kappa}  + (\bm{\alpha}^*)^T \bm{K} \bm{\alpha}^* - 2 \bm{f}_{\mathbb{D}}^T \bm{\alpha}^* \!  \\
		&+\! (N \bar{w}^2)^{-\!1\!} \!\left( \| \bm{\mu}_{\mathbb{D}} \!\!-\! \bm{z} \|^2 \!\!+\!\! \| \bm{w} \|^2 \!\!+\!\! 2 \bm{w}^T\! (\bm{\mu}_{\mathbb{D}} \!-\! \bm{z}) \right) \nonumber\\
        =& \| f \|^2_{\kappa}  + (\bm{\alpha}^*)^T \bm{K} \bm{\alpha}^* - 2 \bm{z}^T \bm{\alpha}^* + 2 \bm{w}^T \bm{\alpha}^* \nonumber \\
        &+ \!\!(N \bar{w}^2)^{-\!1\!} \!\left(\! N^2 \bar{w}^4 \| \bm{\alpha}^* \!\|^2 \!\!+\!\! \| \bm{w} \|^2 \!\!-\!\! 2 N \bar{w}^2  \bm{w}^T \bm{\alpha}^* \right) \nonumber \\
        =& \| f \|^2_{\kappa}  \!-\!\! (\bm{\alpha}^*)^T \! (\bm{K} \!+\! N \bar{w}^2 \bm{I}_N) \bm{\alpha}^* \!\!+\!\! (N \bar{w}^2)^{-\!1\!} \| \bm{w} \|^2 \nonumber \\
		=& \| f \|^2_{\kappa}  \!-\! \bm{z}^T (\bm{K} \!+\! N \bar{w}^2 \bm{I}_N)^{-1} \bm{z} \!+\! (N \bar{w}^2)^{-\!1\!} \| \bm{w} \|^2. \nonumber
	\end{align}
	Furthermore, considering the boundness of $\| f \|^2_{\kappa}$ in \cref{assumption_RKHS_bounded} and $\| \bm{w} \|^2 \!\le\! N \bar{w}^2$, we can directly derive
	\begin{align}
		\| \mu \!-\! f \|_{\kappa_{\mathbb{D}}}^2 &\le B^2 \!-\! \bm{z}^T (\bm{K} \!+\! N \bar{w}^2 \bm{I}_N)^{-1} \bm{z} \!+\! 1 \!=\! \beta^2
	\end{align}
	Finally, considering the Cauchy-Schwarz inequality and the definition of the data-dependent power function, the absolute prediction error at $\bm{x}$ is bounded as
	\begin{align}
		| \mu(\bm{x}) - f(\bm{x}) | &= | \langle \mu(\cdot) - f(\cdot), \kappa_{\mathbb{D}}(\bm{x}, \cdot) \rangle_{\kappa_{\mathbb{D}}} | \\
		&\le \| \mu - f \|_{\kappa_{\mathbb{D}}} \sqrt{\kappa_{\mathbb{D}}(\bm{x}, \bm{x})}, \nonumber
	\end{align}
	which leads to the result in \eqref{eqn_KRR_error_bound}.
\end{proof}
\cref{lemma_KRR_error_bound} provides a deterministic error bound for KRR, where the coefficient $\beta$ can also be formulated as data independence by considering $\beta^2 \le B + 1$.
Moreover, this lemma not only provides the prediction guarantee for safety-critical control discussed in \cref{subsection_control_performance}, but also signifies a robust avenue for performance enhancement through the augmentation of data collection efforts as proven in \cite{vivarelli1998studies}.

\section{Learning-based Safe Control with KRR}

\subsection{Data Acquisition and Analysis} \label{subsection_data_collection}
Due to the partial measurements of the system, the acquisition of training data solely through direct measurements alone proves to be a non-trivial task. Therefore, we introduce auxiliary state variables $\tilde{x}_i(t_k)$, which exhibit a structure akin to that described in \eqref{eqn_dynamics}, formulated as
\begin{align} \label{eqn_tilde_x}
	\tilde{x}_{i+1}(t_k) = T^{-1} \big( \tilde{x}_i(t_{k+1}) - \tilde{x}_i(t_k) \big) , 
\end{align}
for $i = 1, \cdots, n-1$ and $\tilde{x}_1(t_k) = y(t_k)$.
Thus, the concatenated auxiliary variable $\tilde{\bm{x}}(t_k) = [\tilde{x}_1(t_k), \cdots, \tilde{x}_n(t_k)]^T$ is defined in association with the current system state $\boldsymbol{x}(t_k)$. Subsequently, the bounded measurement is calculated as $z(t_k) = \tilde{x}_n(t_k) - u(t_k)$. Utilizing the variables $\tilde{\bm{x}}$ and $z$, the data acquisition strategy is designed in \cref{algorithm_data_collection}. Furthermore, to ensure the safety of the data acquisition process, a predefined safe set $\mathbb{S}$ is defined as a subset of $\mathbb{X}$, such that $\bm{x}(t_{k+1}) \in \mathbb{X}$ if $\bm{x}(t_k) \in \mathbb{S}$ for any control law $u(t_k)$ in $\mathbb{U}$.
\begin{algorithm} [t]
	\caption{Data acquisition}
	\label{algorithm_data_collection}
	\begin{algorithmic} [1]
		\Statex Initialize with $\mathbb{D} = \emptyset$;
		\State Reset the initial state such that $\bm{x}(0) \in \mathbb{S}$ with predefined safety set $\mathbb{S} \subseteq \mathbb{X}$; 
		\State Choose an arbitrary control law $u$ from $\mathbb{U}$, and run experiments until $\bm{x}(k_*) \notin \mathbb{S}$ with $k_* \in \mathbb{N}$;
		\State Collect $\bm{Y} = \{ y(t_k) \}_{k = 0,\cdots,k_*}$ and $\bm{U} = \{ u(t_k) \}_{k=0,\cdots,k_*}$;
		\If{ $k_* \ge n$ }
		\State Calculate $\tilde{\bm{x}}(t_k)$ for $k = 0, \cdots, k_* - n$ from \eqref{eqn_tilde_x};
		\State Calculate $z(t_k) = \tilde{x}_n(t_k) - u(t_k)$ for $k = 0, \cdots, k_* - n$;
		\State $\mathbb{D} \leftarrow \{ \mathbb{D}, \{ \tilde{\bm{x}}(t_k), z(t_k) \}_{k = 0, \cdots, k_* - n} \}$
		\EndIf
		\Statex Repeat Step 1 to 8.
	\end{algorithmic}
\end{algorithm}
\begin{remark}
	\cref{algorithm_data_collection} allows online data collection during the operation.
	However, the data pair obtained using \cref{algorithm_data_collection} is delayed, i.e., at $k \ge n$ only $\{ \tilde{\bm{x}}(t_{k-n}), z(t_{k-n}) \}$ is available and no data pair for $k < n$. 
    It is unlikely to the instant acquirement of data pair $\{ x, z \}$ in continuous-time setting \cite{dai2023can}, which can be used to improve the prediction performance immediately.
% The detailed discussion for $\mathbb{S}$ and $\mathbb{U}$ is beyond the scope of this paper, and is left in future work.
\end{remark}
As the measurement noise $v$ is constrained by $\bar{v}$, it is intuitive to expect that the noise in the data set generated through \cref{algorithm_data_collection} is also bounded. Moreover, \cref{assumption_dataset} requires precise state measurement and transfers the noise into the variable $z$. This noise transformation is challenging due to the presence of an unknown function $f(\cdot)$ characterized by discontinuous behavior. Consequently, the limited variation of $f(\cdot)$ indicates the necessity of Lipschitz continuity of $f(\cdot)$, which is elucidated in the subsequent lemma.
\begin{lemma} [\cite{hashimoto2022learning}] \label{lemma_Lf}
	Let $f(\cdot)$ satisfy \cref{assumption_RKHS_bounded} for a locally Lipschitz kernel $\kappa$ with Lipschitz constant $L_{\kappa} \in \mathbb{R}_+$ induced by the Euclidean norm, i.e., $\sup_{\bm{x} \in \mathbb{X}}\| \nabla \kappa(\bm{x}) \| \!\le\! L_{\kappa}$.
	Then, the function $f(\cdot)$ is also Lipschitz, i.e., $\| \nabla f(\bm{x}) \| \!\le\! L_f, \forall \bm{x} \!\in\! \mathbb{X}$ with Lipschitz constant written as $L_f \!=\! \sqrt{2 L_{\kappa}} B$.
\end{lemma}
Given the established Lipschitz continuity of $f(\cdot)$, the consideration turns to the noisy measurement $z^{(\iota)}$ associated with $f(\cdot)$. Notably, $z^{(\iota)}$ should be obtained via $f(\tilde{\bm{x}}^{(\iota)})$, wherein the auxiliary state $\tilde{\bm{x}}^{(\iota)}$ is utilized, instead of $f(\bm{x}^{(\iota)})$ involving the true state $\bm{x}^{(\iota)}$. Furthermore, it is demonstrated that the upper bound $\bar{w}$ of the measurement noise $w$ for $z$ is directly inherited from $\bar{v}$. The precise expression of this inheritance is expounded upon in the subsequent lemma.
\begin{lemma}
	Consider a $n$-order discrete-time system \eqref{eqn_dynamics} with Lipschitz $f(\cdot)$ and measurement noise on $y$ bounded by $\bar{v}$.
	Using the data collection strategy in \cref{algorithm_data_collection}, then the measurement noise in \cref{assumption_dataset} is bounded by
	\begin{align} \label{eqn_data_w_bar}
		\bar{w} \!=\! \Big( \Big( \frac{2}{T} \Big)^{n-1} + L_f \sqrt{{1 - (2/T)^{2n}}/ {1 - (2/T)^2}} \Big) \bar{v}
	\end{align}
	with the time interval $T$ and $L_f$ in \cref{lemma_Lf}.
\end{lemma}
\begin{proof}
    It is evident that $\tilde{x}_i(t_k)$ constitutes a noisy representation of $x_i(t_k)$, expressed as $\tilde{x}_i(t_k) \!=\! x_i(t_k) \!+\! v_i(t_k)$, $\forall i \!=\! 1, \!\cdots\!, n$, where $v_1(t_k) \!=\! v(t_k)$. 
    Moreover, the noisy $f(\bm{x}(t_k))$ is constructed by $z(t_k) \!=\! \tilde{x}_n(t_k) \!-\! u(t_k)$, such that $v_i(t_k)$ inherited from $v(t_k)$ is iteratively calculated as
	\begin{align}
		v_{i+1}(t_k) &= \tilde{x}_{i+1}(t_k) - x_{i+1}(t_k) \nonumber \\
		&= T^{-1} \big( (\tilde{x}_i(t_{k + 1}) \!-\! x_i(t_{k + 1})) \!-\! (\tilde{x}_i(t_k) \!-\! x_i(t_k)) \big) \nonumber \\
		&= T^{-1} \big( v_i(t_{k + 1}) - v_i(t_k) \big). 
	\end{align}
	Considering the fact that
	\begin{align}
		z(t_k) - f(\bm{x}(t_k)) = v_n(t_k) = T^{1-n} \sum_{i=0}^{n-1} \binom{i}{n-1} v(t_{k+i}), \nonumber
	\end{align}
	then the bound between the collected $z$ and the function $f(\cdot)$ associated with the auxiliary state variable $\tilde{\boldsymbol{x}}$ is derived as
	\begin{align} \label{eqn_data_error_transferr}
		| z(t_k) \!-\! f(\tilde{\bm{x}}(t_k)) | &\!\le\! |v_n(t_k)| + | f(\bm{x}(t_k)) - f(\tilde{\bm{x}}(t_k)) |  \nonumber\\
		&\!\le\! (2 / T)^{n\!-\!1\!} \bar{v} \!+\! L_f \| \bm{x}(t_k) \!-\! \tilde{\bm{x}}(t_k) \|.
	\end{align}
	Note that $\| \bm{x}(t_k) - \tilde{\bm{x}}(t_k) \|$ is bounded as
	\begin{align} \label{eqn_data_x_error}
		\| \bm{x}(t_k) - \tilde{\bm{x}}(t_k) \|^2 & = \sum_{i=1}^n v_i^2(t_k)\le \sum_{i=1}^n (2/T)^{i-1} \bar{v}^2 \\
		&= (1 - (2/T)^{2n}) (1 - (2/T)^2)^{-1} \bar{v}^2. \nonumber
	\end{align}
	Substitute \eqref{eqn_data_x_error} into \eqref{eqn_data_error_transferr}, then \eqref{eqn_data_w_bar} is derived.
\end{proof}

\subsection{Tracking Control with Performance Guarantee} \label{subsection_control_performance}
To achieve the control objective, we propose a learning-based control law with a state observer as
\begin{align} \label{eqn_controller}
	u(t_k) = - \mu(\hat{\bm{x}}(t_k)) + r(t_k) + \bm{\phi}^T (\hat{\bm{x}}(t_k) - \bm{s}(t_k))
\end{align}
with the state observer for estimating $\boldsymbol{x}$ as
\begin{align} \label{eqn_observer}
    \hat{\bm{x}}(t_{k + 1}) =& (\bm{A} + \bm{b} \bm{\phi}^T) (\hat{\bm{x}}(t_k) - \bm{s}(t_k)) + \bm{s}(t_{k+1}) \\
	& + \bm{\theta} (\hat{x}_1(t_k) - y(t_k) ), \nonumber
\end{align}
where $\hat{\bm{x}} = [\hat{x}_1(t_k), \cdots, \hat{x}_n(t_k)]^T$ is the estimated state. The control gain and observer gain denotes $\bm{\phi} \in \mathbb{R}_{>0}^n$ and $\bm{\theta} \in \mathbb{R}_{>0}^n$, respectively. The matrix $\bm{A} $ and vector $\bm{b}$ are written as
\begin{align*}
	&\bm{A} = \begin{bmatrix}
		\bm{I}_{n-1} & \bm{0}_{(n-1) \times 1} \\
		\bm{0}_{1 \times (n-1)} & 0
	\end{bmatrix} + \begin{bmatrix}
		\bm{0}_{(n-1) \times 1} & T \bm{I}_{n-1} \\
		0 & \bm{0}_{1 \times (n-1)}
	\end{bmatrix}, \\
	&\bm{b} = [\bm{0}_{1 \times (n-1)}, 1]^T.
\end{align*}

To evaluate the system performance, the tracking error $\bm{e} \!=\! \bm{x} \!-\! \bm{s}$ and the observation error $\hat{\bm{e}} \!=\! \hat{\bm{x}} \!-\! \bm{x}$ are introduced, whose dynamics combining \eqref{eqn_dynamics}, \eqref{eqn_controller} and \eqref{eqn_observer} are written as
\begin{subequations}\label{eqn_error_dynamics}
\begin{align} 
	\bm{e}(t_{k+1}) =& (\bm{A} + \bm{b} \bm{\phi}^T) \bm{e}(t_k) + \bm{b} \bm{\phi}^T \hat{\bm{e}}(t_k) \nonumber\\
 &+ \bm{b} (f(x(t_k)) - \mu(\hat{x}(t_k))),  \\
	\hat{\bm{e}}(t_{k+1}) =& (\bm{A} + \bm{\theta} \bm{c}^T) \hat{\bm{e}}(t_k) + \bm{\theta} v(t_k)\nonumber \\
 &- \bm{b} (f(x(t_k)) - \mu(\hat{x}(t_k))) , 
\end{align}
\end{subequations}
where $\bm{c} = [1, \bm{0}_{(n-1) \times 1}^T]^T$. Moreover, all eigenvalues of $\bm{A} + \bm{b} \bm{\phi}^T$ and $\bm{A} + \bm{\theta} \bm{c}^T$ lie inside the unit circle by the designed control gain and observer gain. The existence of $\bm{\phi}$ and $\bm{\theta}$ is guaranteed by considering that the pair $(\bm{A}, \bm{b})$ is controllable and $(\bm{A}, \bm{c}^T)$ is observable from the system structure in \eqref{eqn_dynamics}.
With given desired eigenvalues, $\bm{\phi}$ and $\bm{\theta}$ can be obtained by using Ackermann's formula.

Combining the dynamics of tracking and observation error in \eqref{eqn_error_dynamics}, the concatenated tracking error dynamics denotes
\begin{align} \label{eqn_concatenated_error_dynamics}
	\tilde{\bm{e}}(t_{k\!+\!1}) \!\!=\! \tilde{\bm{A}} \tilde{\bm{e}}(t_k) \!+\! \tilde{\bm{b}} (f(x(t_k)) \!\!-\! \mu(\hat{x}(t_k))) \!+\! \tilde{\bm{\theta}} v(t_k),
\end{align}
where $\tilde{\bm{e}}(t_k) = [\bm{e}^T(t_k), \hat{\bm{e}}^T(t_k)]^T$ and
\begin{align*}
	\tilde{\bm{A}} = \begin{bmatrix}
		\bm{A} + \bm{b} \bm{\phi}^T & \bm{b} \bm{\phi}^T \\
		\bm{0}_{n \times n} & \bm{A} + \bm{\theta} \bm{c}^T
	\end{bmatrix}, ~~ \tilde{\bm{b}} = \begin{bmatrix}
		\bm{b} \\ - \bm{b}
	\end{bmatrix},~~ \tilde{\bm{\theta}} = \begin{bmatrix}
		\bm{0}_{n \times 1} \\ \bm{\theta}
	\end{bmatrix}.
\end{align*}
Note that $\tilde{\bm{A}}$ is also a Schur matrix, whose proof is straightforward by considering $\{ \lambda_i(\tilde{\bm{A}}) \}_{i = 1, \cdots, 2n} = \{ \{ \lambda_i(\bm{A} + \bm{b} \bm{\phi}^T) \}_{i = 1, \cdots, n}, \{ \lambda_i(\bm{A} + \bm{\theta} \bm{c}^T) \}_{i = 1, \cdots, n} \}$ due to its block triangular structure.
The control performance by using the proposed controller in \eqref{eqn_controller} and \eqref{eqn_observer} is shown as follows.
	\begin{theorem} \label{theorem_tracking_error_bound}
		Consider a discrete-time system \eqref{eqn_dynamics} satisfying \cref{assumption_RKHS_bounded}, and using the proposed learning-based controller with a state observer in \eqref{eqn_controller} and \eqref{eqn_observer} to track a predefined trajectory \eqref{eqn_reference}.
		The prediction $\mu$ of unknown $f$ used in \eqref{eqn_controller} is obtained via KRR for the optimization problem in \eqref{eqn_KRR_optimization} with a data set satisfying \cref{assumption_dataset}.
		Choose the control and observation gains $\bm{\phi}$, $\bm{\theta}$ and symmetric positive matrix $\bm{Q}$, such that $\xi_0 > 0$ with
		\begin{align}
			\xi_0 = \underline{\lambda}(\bm{Q}) \!-\! 2 \sqrt{2} L_f \| \tilde{\bm{A}}^T \bm{P} \| - 2 L_f^2 \| \bm{P} \|,
		\end{align}
		where $L_f$ is the Lipschitz constant of $f(\cdot)$ from \cref{lemma_Lf} and $\bm{P} \in \mathbb{R}^{2n \times 2n}$ denote the unique solution of discrete-time Lyapunov equation $\tilde{\bm{A}}^T \bm{P} \tilde{\bm{A}} - \bm{P} = - \bm{Q}$ for a given symmetric positive matrix $\bm{Q} \in \mathbb{R}^{2n \times 2n}$.
        The operator $\underline{\lambda}(\cdot)$ returns the minimal eigenvalue of the matrix.
		Then, there exists $\bar{k} \in \mathbb{N}$ such that the tracking error and the observation error are both ultimately bounded, i.e.,
		\begin{align}
			&\| \bm{e}(t_k) \| \le \chi \xi \big( \sqrt{2} \beta  \bar{P} \!+\! \| \bm{\theta} \| \bar{v} \big), \\
			&\| \hat{\bm{e}}(t_k) \| \le \chi \xi \big( \sqrt{2} \beta  \bar{P} \!+\! \| \bm{\theta} \| \bar{v} \big),
		\end{align}
		for $\forall k \ge \bar{k}$, where $\chi = (\bar{\lambda}(\bm{P}) / \underline{\lambda}(\bm{P}))^{1/2}$, $\xi = \xi_0^{-1} \xi_1 + (1 + \xi_0^{-1} \xi_2)^{1/2}$ and $\bar{P} = \sup_{\bm{x} \in \mathbb{X}} P(\bm{x})$ with
		\begin{align*}
			&\xi_1 = \| \tilde{\bm{A}}^T \bm{P} \| + 2 \beta L_f \| \bm{P} \| \bar{P} + \sqrt{2} L_f \| \bm{P} \| \| \bm{\theta} \| \bar{v},
		\end{align*}
		in which $\xi_2 = \| \bm{P} \|$ and $\beta$ is defined in \cref{lemma_KRR_error_bound}.
	\end{theorem}
	\begin{proof}
		The control performance is analyzed using Lyapunov theory, where the Lyapunov candidate is chosen as $V(t_k) = \tilde{\bm{e}}^T(t_k) \bm{P} \tilde{\bm{e}}(t_k)$
		with the symmetric positive definite $\bm{P}$.
		The existence and uniqueness of $\bm{P}$ is proved by considering that $\tilde{\bm{A}}$ is a Schur matrix.
		Applying the concatenated error dynamics in \eqref{eqn_concatenated_error_dynamics}, the difference of Lyapunov function $\Delta V(t_k) = V(t_{k+1}) - V(t_k)$ between two consecutive time instances is written as
		\begin{align}
			&\Delta V(t_k) \!= \!-\! \tilde{\bm{e}}^T\!(t_k) \bm{Q} \tilde{\bm{e}}(t_k)  \\
			&~~~~+ 2 \tilde{\bm{e}}^T\!(t_k) \tilde{\bm{A}}^T\! \bm{P} \big(\tilde{\bm{b}} (f(x(t_k)) \!-\! \mu(\hat{x}(t_k))) \!+\! \tilde{\bm{\theta}} v(t_k) \big) \nonumber\\
			&~~~~+ \big(\tilde{\bm{b}} (f(x(t_k)) \!-\! \mu(\hat{x}(t_k))) \!+\! \tilde{\bm{\theta}} v(t_k) \big)^T \bm{P} \big(\tilde{\bm{\theta}} v(t_k)  \nonumber \\
			& ~~~~~~~~~ \!+\! \tilde{\bm{b}} (f(x(t_k)) \!-\! \mu(\hat{x}(t_k)))  \big). \nonumber
		\end{align}
		Using the properties of eigenvalue and singular value, the difference $\Delta V(t_k)$ is bounded by
		\begin{align}
			\Delta &V(t_k) \!\le - \underline{\lambda}(\bm{Q}) \| \tilde{\bm{e}}(t_k) \|^2 \nonumber \\
			&\!+\! 2 \| \tilde{\bm{A}}^T \!\bm{P} \| \! \| \tilde{\bm{e}}(t_k) \| \!\big(\! \| \tilde{\bm{b}} \| |f(x(t_k)) \!-\! \mu(\hat{x}(t_k))| \!\!+\!\! \| \tilde{\bm{\theta}} \| |v(t_k)| \big) \nonumber \\
			&\!+\! \| \bm{P} \| \big( \| \tilde{\bm{b}} \| |f(x(t_k)) \!-\! \mu(\hat{x}(t_k))| \!+\! \| \tilde{\bm{\theta}} \| |v(t_k)| \big)^2, 
		\end{align}
		where the prediction error bound using estimated system states $\hat{\bm{x}}$ is bounded as
		\begin{align}
			|f(x) \!-\! \mu(\hat{x})| &\!\le\! |f(\bm{x}) \!-\! f(\hat{\bm{x}})| \!+\! |f(\hat{\bm{x}}) \!-\! \mu(\hat{\bm{x}})|  \\
			&\le L_f \| \hat{\bm{e}} \| + \beta P(\hat{\bm{x}}) \le L_f \| \tilde{\bm{e}} \| + \beta P(\hat{\bm{x}}), \nonumber
		\end{align}
		such that $\Delta V(t_k)$ is further bounded as
		\begin{align}
			\Delta V(t_k) \!\le\!& -\! \big(\underline{\lambda}(\bm{Q}) \!-\! 2 L_f \| \tilde{\bm{A}}^T \! \bm{P} \| \| \tilde{\bm{b}} \| \!-\! L_f^2 \| \bm{P} \| \tilde{\bm{b}} \|^2 \big) \| \tilde{\bm{e}}(t_k) \|^2 \nonumber \\
			& \!+\! 2 \big(\| \tilde{\bm{A}}^T \! \bm{P} \| \!+\!    \beta L_f \| \bm{P} \| \| \tilde{\bm{b}} \|^2  P(\hat{\bm{x}}(t_k)) \nonumber\\
			& ~~~~~~~ \!+\! L_f \| \bm{P} \| \| \tilde{\bm{b}} \| \| \tilde{\bm{\theta}} \| |v(t_k)| \big) \| \tilde{\bm{e}}(t_k) \| \\
			& \!+\! \| \bm{P} \| \big( \beta \| \tilde{\bm{b}} \| P(\hat{\bm{x}}(t_k)) \!+\! \| \tilde{\bm{\theta}} \| |v(t_k)| \big)^2  \nonumber\\
			\le& - \xi_0 \| \tilde{\bm{e}}(t_k) \|^2 \!+\! \xi_2 \big( \beta \| \tilde{\bm{b}} \| P(\hat{\bm{x}}(t_k)) \!+\! \| \tilde{\bm{\theta}} \| |v(t_k)| \big)^2 \nonumber \\
			&\!+\! 2 \xi_1 \| \tilde{\bm{e}}(t_k) \| \big( \beta \| \tilde{\bm{b}} \| P(\hat{\bm{x}}(t_k)) \!+\! \| \tilde{\bm{\theta}} \| |v(t_k)| \big)  \nonumber
		\end{align}
		considering $\| \tilde{\bm{b}} \| = \sqrt{2}$ and $\| \tilde{\bm{\theta}} \| = \| \bm{\theta} \|$.
		Recall the assumptions of positive coefficients $\xi_0$, $\xi_1$ and $\xi_2$ in the theorem, which leads the right-hand side as a concave function w.r.t $\| \tilde{\bm{e}}(t_k) \|$, and therefore the positive range of $\| \tilde{\bm{e}}(t_k) \|$ ensuring the negativity of $\Delta V(t_k)$ is written as
		\begin{align} \label{eqn_stability_condition_e_1}
			\| \tilde{\bm{e}}(t_k) \| &\ge \xi_0^{\!-\!1\!} \big( \xi_1 \!+\! \sqrt{ \xi_1^2 \!+\! \xi_0 \xi_2} \big) \big( \beta \| \tilde{b} \| \bar{P} \!+\! \| \tilde{\bm{\theta}} \| \bar{v} \big) \nonumber \\
			&= \xi \big( \sqrt{2} \beta  \bar{P} \!+\! \| \bm{\theta} \| \bar{v} \big),
		\end{align}
		ignoring the invalid negative part of $\| \tilde{\bm{e}}(t_k) \|$.
		Considering $\tilde{\lambda} \| \tilde{\bm{e}} \|^2 \!\!\le\!\! V \!\!\le\!\! \bar{\lambda} \| \tilde{\bm{e}} \|^2$, the decay of the Lyapunov function is indicated using \eqref{eqn_stability_condition_e_1} when $V(t_k) \!\!\ge\!\! \bar{\lambda}(\bm{P}) \xi^2 \bar{P}^2$, resulting $V(t_k)$ is ultimately bounded by $V(t_k) \!\le\! \bar{\lambda}(\bm{P}) \xi^2 \bar{P}^2$ and therefore $\| \tilde{\bm{e}}(t_k) \|^2 \!\!\le\!\! \bar{\lambda}(\bm{P}) / \underline{\lambda}(\bm{P}) \xi^2 \bar{P}^2$.
		The proof is concluded considering $\| \bm{e}(t_k) \| \!\le\! \| \tilde{\bm{e}}(t_k) \|$ and $\| \hat{\bm{e}}(t_k) \| \!\le\! \| \tilde{\bm{e}}(t_k) \|$.
	\end{proof}

	\section{Simulation}

    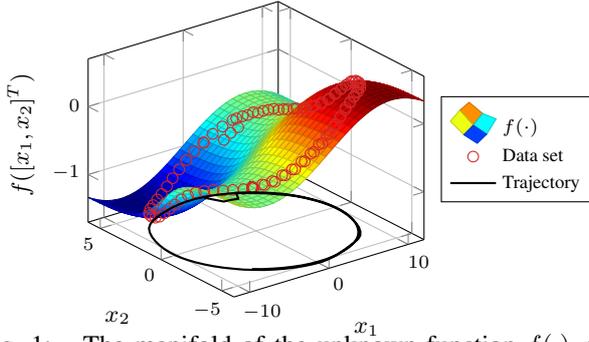
\begin{figure}[t] 
		\centering
		\def\file{figures/DataAndReference.txt}
		\begin{tikzpicture}
			\begin{axis}[xlabel={$x_1$},ylabel={$x_2$},zlabel={$f([x_1,x_2]^T)$},
				view={-37.5}{30},
				xmin=-12, ymin = -6, xmax = 12,ymax=6, zmin=-1.7, zmax=0.7,legend columns=1,
				width=6cm,height=5.5cm,legend style={at={(1.05,0.5)},anchor=west},
				clip mode=individual]
				\addplot3[surf,	opacity=0.3, domain=-12:12,domain y=-6:6,mesh/cols=30,mesh/rows=30, colormap/jet, line width=0.01pt]    table[x = X1_f , y  = X2_f, z = Y_f ]{\file};
				\addplot3[only marks, mark=o, tab10_red]    table[x = X1 , y  = X2, z = Y ]{\file};
				\addplot3[thick, black]    table[x = X1 , y  = X2, z = Y0 ]{\file};
				\legend{$f(\cdot)$, Data set, Trajectory}
			\end{axis}
		\end{tikzpicture}
		\vspace{-0.4cm}
		\caption{
			The manifold of the unknown function $f(\cdot)$, the states of the collected training data set $\mathbb{D}$, and the state trajectory during the data acquisition procedure.
		}
		\vspace{-0.3cm}
		\label{figure_dataset}
	\end{figure}
 
	To demonstrate the effectiveness of the proposed control law \eqref{eqn_controller}, discrete-time dynamics as in \eqref{eqn_dynamics} is considered with $n=2$ and $T = 0.2$.
	The measurement noise of the system output $y$ is bounded by $\bar{v} = 0.01$.
	Moreover, the unknown function is considered as
	$f(\bm{x}) = 0.5 (\sin(0.2 x_1) - 1) + 1 / (1 + \exp (x_2))$ and the kernel function is chosen as $\kappa(\bm{x}, \bm{x}') = \sigma_f^2 \exp(-0.5 l^{-2} \| \bm{x} - \bm{x}' \|^2)$ with $\sigma_f = 0.5$ and $l = 5$ such that the corresponding RKHS norm is bounded by $B = 0.3$.
	The tracking reference trajectory is given with $r(t_k) = 50 ( \sin(0.1k + 0.2) - \sin(0.1k + 0.1))$.
	To achieve such control task, the control gain $\bm{\phi}$ and observation gain $\bm{\theta}$ are determined using Ackermann's formula, such that the eigenvalues of $\bm{A} + \bm{b} \bm{\phi}^T$ and $\bm{A} + \bm{\theta} \bm{c}^T$ are set as $[0.8,0.7]$ and $[0.01,0.02]$, respectively.
	Note that this choice of $\bm{\phi}$ and $\bm{\theta}$ ensures $\xi_0 > 0$ in \cref{theorem_tracking_error_bound}.
	The prediction $\mu$ of the unknown $f$ is predicted using KRR, where the training data set is collected through \cref{algorithm_data_collection} by choosing the controller similar to \eqref{eqn_controller} with the same coefficients.
	In total, $200$ training samples are collected, whose distribution is shown in \cref{figure_dataset}.
	The simulation is initialized at $k = 0$ with $\bm{x}(0) = [0,50\sin(0.1)]^T$ and lasts $200$ steps.

	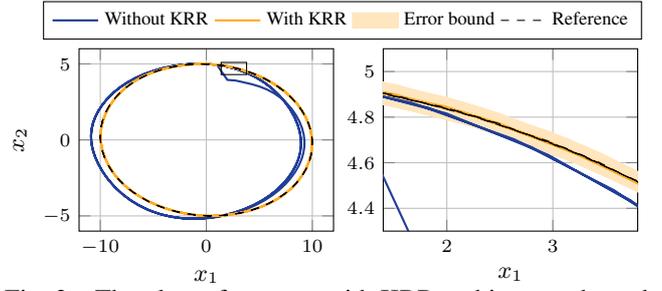
\begin{figure}[t]
		\begin{tikzpicture}
			\def\file{figures/StatePhase.txt}
			\begin{axis}[xlabel={$x_1$},ylabel={$x_2$},
				xmin=-12, ymin = -6, xmax = 12,ymax=6,legend columns=2,
				width=0.28\textwidth,height=4cm,legend pos=north east,
				xshift=-3.5cm,clip mode=individual]
				\addplot[tab10_blue,thick]    table[x = x1_set_no_KRR , y  = x2_set_no_KRR ]{\file};
				\addplot[tab10_orange,thick]    table[x = x1_set_with_KRR , y  = x2_set_with_KRR ]{\file};
				\addplot[black,dashed]    table[x = x1_ref_set , y  = x2_ref_set ]{\file};
				\addplot+[name path=lower,black,no markers, draw=none] table[x = x1_ref_lower_set , y  = x2_ref_lower_set ]{\file};
				\addplot+[name path=upper,black,no markers, draw=none] table[x = x1_ref_upper_set , y  = x2_ref_upper_set ]{\file};
				\addplot[tab10_orange!30] fill between[of=lower and upper];
				\draw[draw=black] (1.4,4.3) rectangle (3.8,5.1);
			\end{axis}
			\begin{axis}[xlabel={$x_1$},ylabel={},
				xmin=1.4, ymin = 4.3, xmax = 3.8,ymax=5.1,legend columns=4,
				width=0.28\textwidth,height=4cm,legend style={at={(-0.16,1.05)},anchor=south},
				xshift=0.5cm,clip mode=individual]
				\addplot[tab10_blue,thick]    table[x = x1_set_no_KRR , y  = x2_set_no_KRR ]{\file};
				\addplot[tab10_orange,thick]    table[x = x1_set_with_KRR , y  = x2_set_with_KRR ]{\file};
				\addplot+[name path=lower,black,no markers, draw=none] table[x = x1_ref_lower_set , y  = x2_ref_lower_set ]{\file};
				\addplot+[name path=upper,black,no markers, draw=none] table[x = x1_ref_upper_set , y  = x2_ref_upper_set ]{\file};
                \addplot[tab10_orange!30] fill between[of=lower and upper];
                \addplot[black,dashed]  table[x = x1_ref_set , y  = x2_ref_set ]{\file};
				\legend{Without KRR, With KRR,,,Error bound, Reference}
			\end{axis}
		\end{tikzpicture}
		\vspace{-0.3cm}
		\caption{
			The plots of states $\bm{x}$ with KRR and its error bound (orange), without KRR (blue), and the desired trajectory (black dashed). The partial magnification plots (right).   
		}
		\vspace{-0.3cm}
		\label{figure_Bound}
	\end{figure}
 
	\begin{figure}[t]
		\centering
		% \vspace{0.3cm}
		\begin{tikzpicture}
			\def\file{figures/Error.txt}
			\begin{semilogyaxis}[xlabel={},ylabel={$\| \bm{e} \|$},
				xmin=0.1, ymin = 1e-5, xmax = 199.9,ymax=9e2,legend columns=3,
				width=0.48\textwidth,height=3.5cm,legend pos=north east,
				xticklabels={,,}]
				\addplot[tab10_blue, thick]    table[x = k , y  = e_no ]{\file};
				\addplot[tab10_orange,thick]    table[x = k , y  = e ]{\file};
				\addplot[tab10_gray,thick]    table[x = k , y  = e_exact ]{\file};
				\legend{Without KRR, With KRR, Exact}
			\end{semilogyaxis}
			\begin{semilogyaxis}[xlabel={$k$},ylabel={$\| \hat{\bm{e}} \|$},
				xmin=0.1, ymin = 1e-5, xmax = 199.9,ymax=3e0,legend columns=2,
				width=0.48\textwidth,height=3.5cm,legend pos=north east,
				yshift=-2.1cm]
				\addplot[tab10_blue, thick]    table[x = k , y  = e_hat_no ]{\file};
				\addplot[tab10_orange,thick]    table[x = k , y  = e_hat ]{\file};
				\addplot[tab10_gray,thick]    table[x = k , y  = e_hat_exact ]{\file};
			\end{semilogyaxis}
		\end{tikzpicture}
		\vspace{-0.3cm}
		\caption{
			The curves of $\| \bm{e}(t) \|$ and $\| \hat{\bm{e}}(t) \|$ over time $t$.
		}
		\vspace{-0.3cm}
		\label{figure_Error}
        \vspace{-0.3cm}
	\end{figure}
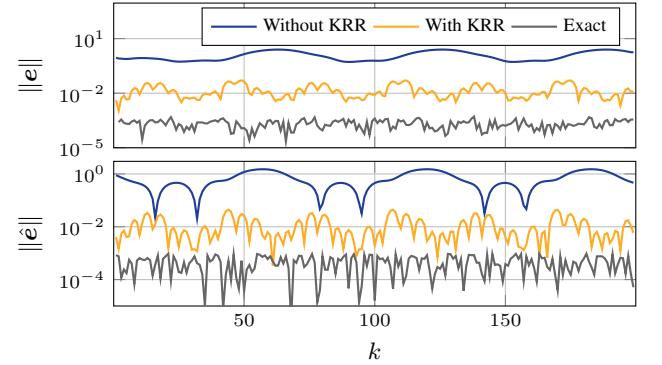

    We conduct a comparative analysis of trajectories both with and without KRR, representing the tracking performance in the state position in \cref{figure_Bound}, where $\mu(\cdot) = 0$ for the case without learning. Moreover, \cref{figure_Error} shows the norm of tracking error $\bm{e}$ and observation error $\hat{\bm{e}}$ for with/without methods and a controller with exact dynamics. Notably, the controller incorporating KRR demonstrates a substantial advancement, manifesting a remarkable $10$ times reduction in both tracking and estimation errors when compared to the controller operating without learning techniques. 

\section{Conclusion}
In summary, this paper presents a novel kernel-based control method for the safe tracking of high-order systems with unknown dynamics and partially observable states. 
Leveraging KRR and a tailored data collection strategy, the proposed approach effectively models system uncertainties with limited state measurements. 
Our proposed learning-based control law integrated with a state observer ensures the deterministic bounds for both tracking and observation errors and underscores its adaptability in safety-critical control tasks. 
The simulations show that there is a significant improvement in the tracking performance, which validates the effectiveness of the proposed method. 

\bibliographystyle{IEEEtran}
\bibliography{refs}

\end{document}